\documentclass[ aps,
                pra,
                showpacs,
                amssymb,
                nofootinbib,
                superscriptaddress,
                twocolumn,
                longbibliography,noeprint, 10pt]{revtex4-2}

\usepackage[utf8]{inputenc}
\usepackage{amsmath,amssymb,amsthm,mathtools}
\usepackage{booktabs,array,enumitem}
\usepackage[expansion=false]{microtype}
\usepackage[dvipsnames]{xcolor}
\usepackage[colorlinks=true,linkcolor=MidnightBlue,citecolor=MidnightBlue,urlcolor=MidnightBlue]{hyperref}

\newtheorem{theorem}{Theorem}

\newtheorem{lemma}{Lemma}
\newtheorem{corollary}{Corollary}
\theoremstyle{definition}

\theoremstyle{remark}

\newcommand{\spec}{\operatorname{spec}}
\newcommand{\Tr}{\operatorname{Tr}}
\newcommand{\N}{\mathbb{N}}

\newcommand{\proj}[1]{|#1\rangle\langle #1|}

\newcommand{\one}{\mathbf 1}

\begin{document}

\title{Dimension-free marginal inequalities for pure three-party states}

\author{Simon Morelli}
\email{simon.morelli@tuwien.ac.at}
\affiliation{Technische Universit{\"a}t Wien, Atominstitut \& Vienna Center for Quantum Science and Technology (VCQ),  Stadionallee 2, 1020 Vienna, Austria}

\begin{abstract}
The quantum marginal problem, concerned with the compatibility of local and global properties of a multipartite quantum system, is arguably one of the fundamental questions in quantum theory. In this article we introduce a novel family of inequalities bounding the allowed local spectra of a pure tripartite state. The number of inequalities in this family grows with the dimension of the system. Most importantly, the inequalities hold for arbitrary dimensions. For equal local dimension at most 4, they generate the complete set of eigenvalue inequalities.
\end{abstract}

\maketitle

\section{Introduction}

The quantum marginal problem (QMP) is arguably a central problem in quantum information. It asks the fundamental question of the compatibility of the parts of a multipartite system~\cite{bravyi2003requirementscompatibilitylocalmultipartite,higuchi2003oneparticlereduceddensitymatrices,klyachko2004quantummarginalproblemrepresentations,YuSimacherWyderkaNguyenGuhne21,Huber_2025}.
Beside its broad theoretical importance, the QMP is also of practical relevance. Often experiments and physical models only give access to parts of a complex composite quantum system. At the same time states are often uniquely defined by their one and two-body reduced marginals~\cite{LindenPopescuWootters_2002,Wyderka_2017}. In many-body systems in condensed matter, e.g. for spin models or interacting electrons in solids, the physically relevant quantities, such as mean values of observables, susceptibilities or correlation functions, typically depend on the properties of the one- and two-body reduced density matrices.
These do not directly depend on the Hamiltonian determining the system, which defines global states of the system, e.g., the ground state. 
The marginal problem hence connects the local physical properties that are compatible with peculiarities of global phases, including magnetism, superconductivity or topologically protected states and sits at the centre of the investigation of correlations in statistical physics and condensed matter systems~\cite{Liu_2007,BenavidesRiverosSchilling2020,CastilloSchilling21}.

Therefore it is no surprise that a case of special interest is the pure QMP, asking which marginals are compatible with a pure global state. 
Focusing on single party marginals, it is sufficient to look at the eigenvalues of the local states.
In this article we consider the simplest non-trivial problem of 3 parties.
In principle, this case was solved  by Klyachko~\cite{klyachko2004quantummarginalproblemrepresentations}. For a fixed dimension the admissible spectra form a convex polytope whose facets come from Schubert calculus. However, in practice it quickly becomes unfeasible to calculate the linear inequalities that govern the allowed spectra and a full characterisation of the set is only known for small dimensions~\cite{Vergne_2017,bulois2025algorithmcomputekroneckercone}.
This is where our result becomes relevant. 
It gives a constructive way of generating such inequalities. Even if these inequalities do not completely characterise the allowed spectra, they give an outer approximation.

In the next section~\ref{sec:algebra} we briefly introduce the necessary algebraic structures, before reviewing the known results of the tripartite pure QMP in section~\ref{sec:review}. We then look at the structure of the known inequalities in section~\ref{sec:general inequalities} and present our main contribution in section~\ref{sec:GHGK}. Finally, in section~\ref{sec:d4} we compare the introduced inequalities with the known inequalities in dimension 4 and 5.

\section{The algebraic framework}\label{sec:algebra}

Two eigenvalue problems are often compared in the literature. \emph{Horn's problem} asks which spectra $A$, $B$ and $A+B$ can have, it is governed by Littlewood--Richardson coefficients and completely solved. The tripartite pure \emph{quantum marginal problem} asks which spectra the marginals of a pure global state can have. It is governed by Kronecker coefficients and solved explicitly only in small dimensions.

\subsection{Partitions and representations}\label{ssec:part and rep}
 
A partition of an integer $n$, $\lambda\vdash n$, is a weakly decreasing sequence of nonnegative integers $\lambda_1\ge\lambda_2\ge\dots$ with $|\lambda|:=\sum_i\lambda_i=n$.
It is drawn as a Young diagram with $\lambda_i$ boxes in row $i$. Its length $\ell(\lambda)$ is the number of rows, and $\lambda_1$ is the first row, also referred to as the width of the diagram. 
We write $\lambda\subseteq\nu$ if $\lambda_i\le\nu_i$ for all $i$ and the boxes of $\nu$ not in $\lambda$ form the skew diagram $\nu/\lambda$.
We write $\lambda\cap\mu:=(\min(\lambda_i,\mu_i))_i$ and $\hat\lambda:=(\lambda_2,\lambda_3,\dots)$. 
 
Denote by $\pi_\lambda$ the irreducible representations of the symmetric group $S_n$ with characters $\chi^\lambda$ and the irreducible polynomial representations of $GL(d)$ by $V^d_\lambda$, $\ell(\lambda)\le d$. By Schur--Weyl duality
\begin{align}
    (\mathbb C^d)^{\otimes n}\;\cong\bigoplus_{\lambda\vdash n,\ \ell(\lambda)\le d}V^d_\lambda\otimes\pi_\lambda,
\end{align}
with $GL(d)$ acting diagonally and $S_n$ permuting the factors. 
The tensor product of two representations of $GL(d)$ corresponding $\lambda$ and $\mu$ is a combination of representations of $GL(d)$
\begin{align}
    V^d_\lambda\otimes V^d_\mu\;\cong\;\bigoplus_\nu c^\nu_{\lambda\mu}\,V^d_\nu,
\end{align}
where the multiplicities $c^\nu_{\lambda\mu}$, with $|\nu|=|\lambda|+|\mu|$, are the \emph{Littlewood--Richardson coefficients}~\cite{fulton2000eigenvaluesinvariantfactorshighest}.
Similarly, the tensor product of two representations $\pi_\lambda$ and $\pi_\mu$ of $S_n$ is a combination of representations of $S_n$
\begin{align}
    \pi_\lambda\otimes\pi_\mu\;\cong\;\bigoplus_{\nu\vdash n}g(\lambda,\mu,\nu)\,\pi_\nu,
\end{align}
with \emph{Kronecker coefficients}
\begin{align}
    g(\lambda,\mu,\nu)=\frac1{n!}\sum_{\sigma\in S_n}\chi^\lambda(\sigma)\chi^\mu(\sigma)\chi^\nu(\sigma).
\end{align}
So $g$ is symmetric in its 3 arguments, and $\pi_\nu$ is a constituent of $\pi_\lambda\otimes\pi_\mu$ if $g(\lambda,\mu,\nu)\neq0$.

\subsection{Horn's problem}\label{ssec:horn}

Let $A,B$ be Hermitian $d\times d$ matrices with spectra $a,b$, and let $c$ be the spectrum of $C=A+B$. Horn's problem asks which triples $(a,b,c)$ occur.
The following result becomes useful later.

\begin{lemma}\label{lem:Weyl}
Let $A,B\ge0$ act on a finite-dimensional space, $C=A+B$, with decreasing zero-padded spectra $a,b,c$. Then~\cite{Bhatia97}
\begin{enumerate}[label=(\roman*),leftmargin=*,itemsep=1pt]
\item $c_i\ge a_i$ and $c_i\ge b_i$ (monotonicity)
\item $c_{i+j-1}\le a_i+b_j$ (Weyl's inequality).
\end{enumerate}
\end{lemma}

For $I=\{i_r>\dots>i_1\}\subset\N$ put $\lambda(I):=(i_r-r,\dots,i_1-1)$ and $v_I:=\sum_{i\in I}v_i$. A triple $(I,J,K)$ of $r$-element sets is a Horn triple if $c^{\lambda(K)}_{\lambda(I)\lambda(J)}\neq0$. Importantly, this condition does not involve $d$.

\begin{theorem}[Klyachko, Knutson--Tao]\label{thm:horn}
\leavevmode
\begin{enumerate}[label=(\roman*),itemsep=1pt]
\item $(a,b,c)$ occurs iff $\sum_ic_i=\sum_ia_i+\sum_ib_i$ and $c_K\le a_I+b_J$ for all Horn triples with $r<d$ \cite{Klyachko98,knutson1999honeycombmodelglntensor,fulton2000eigenvaluesinvariantfactorshighest}.
\item For partitions $\lambda,\mu,\nu$ of length at most $d$, $c^\nu_{\lambda\mu}\neq0$ iff $(\lambda,\mu,\nu)$ occurs as $(a,b,c)$~\cite{fulton2000eigenvaluesinvariantfactorshighest}.
\end{enumerate}
\end{theorem}

The implication ``$(\lambda,\mu,\nu)$ occurs $\Rightarrow c^\nu_{\lambda\mu}\neq0$'' rests on saturation, if $c^{N\nu}_{N\lambda,N\mu}\neq0$ for some $N\ge1$, then $c^\nu_{\lambda\mu}\neq0$ \cite{knutson1999honeycombmodelglntensor}. So Horn's problem is governed exactly by the support of the Littlewood--Richardson coefficients. Horn's recursion \cite{Horn62} generates the Horn triples from integer inequalities on the index sets alone. Moreover, Littlewood--Richardson positivity can be decided in polynomial time \cite{mulmuley2005geometriccomplexityiiideciding,bürgisser2013decidingpositivitylittlewoodrichardsoncoefficients}.

For $r=1$ the Horn triples are $(\{i\},\{j\},\{i+j-1\})$ (Pieri's rule). They give Weyl's inequalities $c_{i+j-1}\le a_i+b_j$, (ii) in Lemma~\ref{lem:Weyl}.
The Horn triple $(K,[r],K)$ gives Lidskii's inequality $c_K\le a_K+b_{[r]}$ and $([r],[r],[r])$ results in Ky Fan's inequality~\cite{Bhatia97}.

\subsection{The quantum marginal problem}\label{ssec:qmp}

Let $\psi\in\mathbb C^{d_A}\otimes\mathbb C^{d_B}\otimes\mathbb C^{d_C}$ be a pure state of arbitrary finite dimensions with local spectra $(\spec\rho_A,\spec\rho_B,\spec\rho_C)$ . The three-party, pure QMP asks which triples $(a,b,c)$ occur, such triples are called admissible. Equivalently, the problem can be reformulated as a question about the admissible triples $(\spec\rho_A,\spec\rho_B,\spec\rho_{AB})$ of a two-party mixed state, since $\rho_{AB}$ and $\rho_C$ have the same non-zero spectrum. For fixed local dimensions the admissible triples form a convex polytope, an instance of a moment polytope \cite{klyachko2004quantummarginalproblemrepresentations}. Its facets are its irredundant defining inequalities, facets that differ only by a permutation of the parties form a class.

Since $\lambda/n$ is a decreasing probability vector, partitions serve as integer approximations of spectra. By Schur--Weyl duality, $g$ gives the branching from a composite system to its parts \cite{ressayre2019horninequalitiesnonzerokronecker,Walter14}
\begin{align}
    V^{ab}_\nu\big|_{GL(a)\times GL(b)}\;\cong\bigoplus_{\lambda,\mu}g(\lambda,\mu,\nu)\,V^a_\lambda\otimes V^b_\mu .
\end{align}
For a density matrix $\rho$, the weight of $\rho^{\otimes n}$ in the block labelled $\lambda$ concentrates on $\lambda/n\approx\spec\rho$ as $n\to\infty$. This gives rise to the following result.
 
\begin{theorem}[Klyachko; Christandl--Mitchison; Christandl--Harrow--Mitchison]\label{thm:qmp-kron}
\leavevmode
\begin{enumerate}[label=(\roman*),itemsep=1pt]
\item If $g(\lambda,\mu,\nu)\neq0$, then $(\lambda,\mu,\nu)/n$ is admissible \cite{Christandl_2007}.
\item Every admissible triple is a limit of such normalized triples \cite{Christandl_2005}. The rational admissible triples are exactly the triples $(\lambda,\mu,\nu)/n$ with $g(\lambda,\mu,\nu)\neq0$ \cite{klyachko2004quantummarginalproblemrepresentations}.
\end{enumerate}
\end{theorem}
 
Unlike Littlewood--Richardson coefficients, Kronecker coefficients are not saturated. $[1,1]$ is the sign representation, so $g((1,1),(1,1),(1,1))=0$, whereas $g((2,2),(2,2),(2,2))=1$~\cite{Klyachko_2006}. The triple with all 3 marginals $(\tfrac12,\tfrac12)$ is admissible and corresponds to the GHZ state. So the quantum marginal polytope only sees the asymptotic support of $g$. Deciding $g\neq0$ is NP-hard~\cite{Ikenmeyer_2017}.

The following 2 results play an important role in this article.

\begin{theorem}[First-row theorem; Klemm, Dvir, Clausen--Meier]\label{thm:firstrow}
Let $\mu,\nu\vdash n$ and $\mu\cap\nu:=(\min(\mu_i,\nu_i))_i$.
\begin{enumerate}[label=(\roman*),itemsep=1pt]
\item $\max\{\lambda_1:\ g(\lambda,\mu,\nu)\neq0\}=|\mu\cap\nu|$~\cite{Klemm77,Dvir93,Clausen93}.
\item If $\lambda\vdash n$ with $\lambda_1=|\mu\cap\nu|$, then~\cite{Dvir93}
\begin{align}
    g(\lambda,\mu,\nu)=\bigl\langle\chi^{\mu/(\mu\cap\nu)}\,\chi^{\nu/(\mu\cap\nu)},\,\chi^{\hat\lambda}\bigr\rangle .
\end{align}
\end{enumerate}
\end{theorem}

Dvir's recursion expresses Kronecker coefficients through Littlewood--Richardson coefficients.

\begin{theorem}[Dvir's recursion]\label{thm:dvir}
Let $\lambda,\mu,\nu\vdash n$, then~\cite{Dvir93,regev2010kroneckermultiplicitieskellhook}
\begin{align}
    g(\lambda,\mu,\nu)&\le\sum_{\substack{\alpha\vdash\lambda_1\\ \alpha\subseteq\mu\cap\nu}}\bigl\langle\chi^{\mu/\alpha}\chi^{\nu/\alpha},\chi^{\hat\lambda}\bigr\rangle\notag\\
    &=\sum_{\alpha\vdash\lambda_1}\sum_{\rho,\tau\vdash n-\lambda_1} c^\mu_{\alpha\rho}\,c^\nu_{\alpha\tau}\,g(\rho,\hat\lambda,\tau).
\end{align}
The equality follows from the definition of the skew character $\chi^{\nu/\alpha}:=\sum_\rho c^\nu_{\alpha\rho}\chi^\rho$~\cite{Bessenrodt_2014}.
\end{theorem}

The original result by Dvir~\cite{Dvir93} gives an exact recursion, but by omitting the subtraction of a positive term we get the inequality~\cite{regev2010kroneckermultiplicitieskellhook}.

To summarise, Littlewood--Richardson coefficients describe how irreducible representations add, the analogue of $A+B$. Kronecker coefficients describe how an irreducible representation of a composite system splits into local ones, the analogue of $\rho_{AB}\mapsto(\rho_A,\rho_B)$.

\section{Known inequalities for the QMP}\label{sec:review}

In the following we focus on the pure state QMP for 3 parties. 
Let $\psi\in\mathcal H_A\otimes\mathcal H_B\otimes\mathcal H_C$ be a pure state with $\|\psi\|^2=m$ of arbitrary finite local dimensions. The local spectra $a,b,c$ are ordered decreasingly, i.e. $x_1\ge x_2\ge\dots\ge0$ for $x=a,b,c$.
For fixed local dimensions $d_A$, $d_B$ and $d_C$, the set of admissible spectral triples is a convex polytope, which is called the marginal polytope~\cite{klyachko2004quantummarginalproblemrepresentations,Christandl_2007}.

For 3 qubits the admissible local spectra are characterized by the polygon inequality
\begin{align}
    a_2\le b_2+c_2\label{qubit1}
\end{align}
and permutations of the parties~\cite{Higuchi_2003,bravyi2003requirementscompatibilitylocalmultipartite}. 

For 3 qutrits, Higuchi~\cite{higuchi2003oneparticlereduceddensitymatrices} proved that the following 7 inequalities, plus permutations of the party labels, are necessary and sufficient. With decreasing order of the spectra they read
\begin{align}
a_2+a_3&\le b_2+b_3+c_2+c_3,\label{H1}\\ 
a_1+a_3&\le b_2+b_3+c_1+c_3,\label{H2}\\
a_1+a_2&\le b_2+b_3+c_1+c_2,\label{H3}\\
2a_2+a_3&\le 2b_2+b_3+2c_2+c_3,\label{H4}\\
2a_3+a_2&\le 2b_2+b_3+2c_3+c_2,\label{H5}\\
2a_2+a_1&\le 2b_2+b_3+2c_2+c_1,\label{H6}\\
2a_2+a_1&\le 2b_3+b_2+2c_1+c_2.\label{H7}
\end{align}
These inequalities define a polytope with 33 vertices and 45 facets, which had been obtained earlier by Franz \cite[Prop.~5.1]{Franz2002} via geometric invariant theory and a computer check. Higuchi proves Eq.~\eqref{H1} via a polygon-type theorem, Eqn.~\eqref{H2}-\eqref{H3} algebraically, and Eqn.~\eqref{H4}-\eqref{H7} variationally~\cite{higuchi2003oneparticlereduceddensitymatrices}.

In every dimension , the quantum marginal polytope is described by finitely many inequalities from Schubert calculus \cite{klyachko2004quantummarginalproblemrepresentations}. Explicit minimal lists exist only for small formats. Vergne--Walter \cite{Vergne_2017} give $4\times4\times4$: 270 facets in 50 classes under permutations of the parties, and 328 extreme rays in 65 classes. Bulois--Denis--Ressayre \cite{bulois2025algorithmcomputekroneckercone} computed the Kronecker cones $\mathrm{Kron}(d,d,d)$ for $d\le7$: up to $S_3$ and excluding the ordering (dominance) inequalities, they find 47, 462, 5\,333 and 64\,792 classes for $d=4,5,6,7$.

Aside from these lists of inequalities for fixed dimensions, there exist few results valid in all dimensions.
Klyachko's maximal-eigenvalue theorem \cite{klyachko2004quantummarginalproblemrepresentations} contains the inequality 
\begin{align}
    a_1\le m + b_i - c_i\label{K1}
\end{align}
valid for every $i$ for all local dimensions.
Other than this, Ressayre~\cite{ressayre2019horninequalitiesnonzerokronecker} proves a general family of inequalities that is referred to at a later point in the article.

\section{General families of inequalities}\label{sec:general inequalities}

At first sight, inequality~\eqref{qubit1} appears different from the inequalities~\eqref{H1}-\eqref{H7}. However, after a little reformulation, they can be brought to the same form. Using $x_2=m-x_1$, Eq.~\eqref{qubit1} can be rewritten in the form of Eq.~\eqref{K1}. Similarly, Eqn.~\eqref{H1}-\eqref{H3} can be rewritten in the same form, for $i=1,2,3$.
These inequalities, proven explicitly for dimension $d=2$~\cite{Higuchi_2003} and $d=3$~\cite{higuchi2003oneparticlereduceddensitymatrices}, follow from the general inequality~\eqref{K1} that is valid in every dimension.

\begin{lemma}[Klyachko~{\cite[Thm.~6.3.1]{klyachko2004quantummarginalproblemrepresentations}}]\label{lem:Klyachko}
Let $\psi$ be a pure state with $\|\psi\|^2=m$ and marginal spectra $(a,b,c)$. Then
\begin{align}\label{eq:first_row}
    a_1\le\sum_i\min(b_i,c_i).   
\end{align}
As a direct consequence,
\begin{align}\label{eq:seed_1}
    a_1\le m+b_I-c_I
\end{align}
for every finite $I\subset\N$, where $x_I=\sum_{i\in I}x_i$.
\end{lemma}

Since its proof is insightful we include it.

\begin{proof}
The non-zero spectra of $\rho_A$ and $\rho_{BC}$ coincide, so $a_1$ is the largest eigenvalue of $\rho_{BC}$. Let $\chi$ be a corresponding unit eigenvector, so $\rho_{BC}\ge a_1\proj\chi$. Taking partial traces gives $\rho_B\ge a_1\,\Tr_C\proj\chi$ and $\rho_C\ge a_1\,\Tr_B\proj\chi$. Both operators $Tr_C\proj\chi$ and $Tr_B\proj\chi$ have the same spectrum $s$ with $\sum_i s_i=1$. By Weyl monotonicity in Lemma~\ref{lem:Weyl}, $b_i\ge a_1 s_i$ and $c_i\ge a_1 s_i$. Summing $\min(b_i,c_i)\ge a_1 s_i$ over $i$ gives the claim. Bound $\min(b_i,c_i)$ by $b_i$ for $i\in I$ and by $c_i$ otherwise and $a_1\le b_I+m-c_I$ follows.
\end{proof}

Alternatively, Klyachko's proof of Thm.~6.3.1~\cite{klyachko2004quantummarginalproblemrepresentations} is based on the first-row Theorem~\ref{thm:firstrow} by Klemm, Dvir and Clausen--Meier.
It states that the maximal first row among the constituents of $\pi_\mu\otimes\pi_\nu$ is $|\mu\cap\nu|=\sum_i\min(\mu_i,\nu_i)$. Hence
\begin{equation}
g(\mu,\lambda,\nu)\ne0\ \Longrightarrow\ \lambda_1\le\sum_i\min(\mu_i,\nu_i) .
\end{equation}
By the Kronecker--spectrum correspondence Theorem~\ref{thm:qmp-kron} this is Lemma~\ref{lem:Klyachko}.
For $i=1$, Eq.~\eqref{K1} is Murnaghan's depth triangle~\cite{klyachko2004quantummarginalproblemrepresentations,ressayre2019horninequalitiesnonzerokronecker}.

Let us now turn to the remaining Higuchi inequalities. The inequalities~\eqref{H4}-\eqref{H6} can be reformulated to
\begin{align}
    a_1-a_2\le m-b_i+c_i+b_j-c_j\label{G2}
\end{align}
for all $i,j$ and finally inequality~\eqref{H7} becomes
\begin{align}
    a_1-a_{i+1}\le m-b_{i+j-1}+b_k+c_j-c_k\label{G3}
\end{align}
for all $i,j,k$.
Somewhat surprisingly, in this form these inequalities also hold in dimension 4 and 5. This is also the reason behind choosing this specific notation for Eq.~\eqref{G3}. We now introduce our first result.

\begin{theorem}\label{thm:G23}
Eq.~\eqref{G2} and Eq.~\eqref{G3} hold for pure states of arbitrary local dimensions, for every permutations of the party labels. For partitions with $g\neq0$ they hold with $m$ replaced by $n$.
\end{theorem}

We omit the proof for now and will see later, that Theorem~\ref{thm:G23} follows from a more general result.
For $d=3$, Eqn.~\eqref{K1},~\eqref{G2} and~\eqref{G3} is Franz's list \cite{Franz2002} and Higuchi's \cite{higuchi2003oneparticlereduceddensitymatrices}. For all $d$, Eq.~\eqref{G2} with $i=1$ follows from Ressayre's Theorem~2~\cite{ressayre2019horninequalitiesnonzerokronecker}.
To the best of our knowledge, Eq.~\eqref{G2} for $i\ge2$ and Eq.~\eqref{G3} for general dimension is a new result.

\section{A new family of inequalities}\label{sec:GHGK}

Before stating our main contribution, we first adapt a result from~\cite{Morelli_2024}.

\begin{lemma}[Siewert's recursion]\label{lem:peeling}
Let $\psi\in\mathcal H_A\otimes\mathcal H_B\otimes\mathcal H_C$, $\|\psi\|^2=m$, and let $a=\spec\rho_A$. Let
\begin{align}
    \psi=\sum_j\sqrt{a_j}\,e_j\otimes\varphi_j
\end{align}
be the Schmidt-decomposition across $A|BC$. Put
\begin{align}
    P&=a_1\Tr_C\proj{\varphi_1}\\
    P'&=a_1\Tr_B\proj{\varphi_1}\\
    \psi'&=\sum_{j\ge2}\sqrt{a_j}\,e_j\otimes\varphi_j
\end{align}
Then $\rho_B=P+\rho'_B$ and $\rho_C=P'+\rho'_C$, where $\rho'$ denotes the marginals of $\psi'$. Moreover, $P,P'\ge0$ have the same spectrum $\alpha$ with $|\alpha|=a_1$, and $\psi'$ has $\|\psi'\|^2=m-a_1$ and $A$-spectrum $\hat a$.
\end{lemma}

\begin{proof}
Since $\langle e_{j'}|e_j\rangle=\delta_{jj'}$, $\rho_B=\sum_ja_j\Tr_C\proj{\varphi_j}$, and the terms with $j\ge2$ form the $B$-marginal of $\psi'$ and likewise for $C$. $P$ and $P'$ are the 2 marginals of the bipartite pure vector $\sqrt{a_1}\varphi_1$, so they have equal spectra.
\end{proof}

We call $\psi'$ the child and write $\beta=\spec\rho'_B$, $\gamma=\spec\rho'_C$. So the child has spectra $(\hat a,\beta,\gamma)$ and mass $m'=m-a_1$. It is the spectral counterpart of one step of Dvir's recursion Theorem~\ref{thm:dvir}.

Let us now introduce the main result.
For a finite multiset $L$ of positive integers, let $\one_L\in\mathbb Z^{\N}$ be its counting vector, so that $v_L=\langle\one_L,v\rangle$ and write $\one_i:=\one_{\{i\}}$. For $u\in\mathbb Z_{\ge0}^{\N}$ with finite support, the level sets $J_r(u):=\{i:u_i\ge r\}$, $1\le r\le\max u$, satisfy $J_1(u)\supseteq J_2(u)\supseteq\cdots$ and $u=\sum_r\one_{J_r(u)}$. A seed is a linear inequality
\begin{equation}
    \langle u,\rho\rangle+\langle v,\tau\rangle\ \le\ \kappa\,m+\langle w,\eta\rangle ,
\end{equation}
with finitely supported $u,v,w$ and a constant $\kappa$, that holds for all pure states, where $m$ is the mass and $(\rho,\eta,\tau)$ are the spectra. The first-row bound of Eq.~\ref{eq:seed_1} in Lemma~\ref{lem:Klyachko} is the basic example. The theorem below peels $a$, bounds the $b$- and $c$-terms by Horn inequalities whose $J$-sets add up to $u$ and $v$, and applies the seed to the child $(\hat a,\beta,\gamma)$.
 
\begin{theorem}[lifting a seed]\label{thm:lift}
Let
\begin{equation}\label{eq:seed}
    \langle u,\beta\rangle+\langle v,\gamma\rangle\ \le\ \kappa\,(m-a_1)+\langle w,\hat a\rangle ,
\end{equation}
be a seed with $u,v\ge0$ integral. Write $u=\sum_{r=1}^{h}\one_{J_r}$ and $v=\sum_{r=1}^{h'}\one_{J'_r}$ with finite sets $J_r,J'_r$ and let $(I_r,J_r,K_r)$ and $(I'_r,J'_r,K'_r)$ be Horn triples. Let $L,L'$ be finite multisets and $\bar\kappa\ge0$ such that
\begin{align}
    \delta=\sum_{r=1}^{h}\one_{I_r}+\sum_{r=1}^{h'}\one_{I'_r}-\one_L-\one_{L'}
\end{align}
satisfies $\sum_{i=1}^k \delta_i\le \bar\kappa\,k$ for all $k\ge1$.
Then, for all pure states and with $m$ and local spectra $(a,b,c)$ as in Section~\ref{sec:review},
\begin{equation}\label{eq:FL}
(\kappa-\bar\kappa)\,a_1-\langle w,\hat a\rangle+\sum_{r=1}^{h}b_{K_r}-b_L+\sum_{r=1}^{h'}c_{K'_r}-c_{L'}\ \le\ \kappa\,m.
\end{equation}
\end{theorem}
 
\begin{proof}
Peel $a$ according to Lemma~\ref{lem:peeling}, then $b=\spec(P+\rho'_B)$ and $c=\spec(P'+\rho'_C)$ with $\spec P=\spec P'=\tau$ and $|\tau|=a_1$. The child has spectra $(\hat a,\beta,\gamma)$ and mass $m-a_1$.
By Theorem~\ref{thm:horn}, $b_{K_r}\le\tau_{I_r}+\beta_{J_r}$ and $c_{K'_r}\le\tau_{I'_r}+\gamma_{J'_r}$ for each $r$. Summing gives
\begin{align}
    \sum_r b_{K_r}&\le\sum_r\tau_{I_r}+\langle u,\beta\rangle\\
    \sum_r c_{K'_r}&\le\sum_r\tau_{I'_r}+\langle v,\gamma\rangle\ .
\end{align}
By monotonicity in Lemma~\ref{lem:Weyl}, $b_L\ge\tau_L$ and $c_{L'}\ge\tau_{L'}$.
So the $b$- and $c$-terms of Eq.~\eqref{eq:FL} are at most
\begin{align}
    \langle \delta,\tau\rangle+\langle u,\beta\rangle+\langle v,\gamma\rangle.
\end{align}
Since $\tau$ is decreasing with finite support, Abel summation gives
\begin{align}
    \langle \delta,\tau\rangle&=\sum_{k\ge1}(\tau_k-\tau_{k+1})\sum_{j=1}^k w_{j}\notag\\
    &\le \bar\kappa\sum_{k\ge1}k\,(\tau_k-\tau_{k+1})=\bar\kappa|\tau|=\bar\kappa\,a_1 .
\end{align}
Finally, the seed for the child gives
\begin{align}
    \langle u,\beta\rangle+\langle v,\gamma\rangle\le \kappa\,(m-a_1)+\langle w,\hat a\rangle.
\end{align}
Adding up proves Eq.~\eqref{eq:FL}.
\end{proof}
 
Every inequality given by Theorem~\ref{thm:lift} holds in every dimension and, by Theorem~\ref{thm:qmp-kron} (i), for all $\lambda,\mu,\nu\vdash n$ with $g(\lambda,\mu,\nu)\ne0$ and $m\rightarrow n$ in the Kronecker version. A lift can serve as a seed again. Lifting Eq.~\eqref{eq:seed_1} gives one-step inequalities such as the next corollary. Recursively, these can be used to obtain two-step inequalities and so on.

Suppose the $J_r$ are the level sets $J_r(u)$, $I_1\supseteq I_2\supseteq\cdots$ and $K_1\supseteq K_2\supseteq\cdots$. Then the triples $(I_r,J_r,K_r)_{r\le h}$ form Horn triples of length $h=\max u$ and likewise for $v$. 

\begin{corollary}\label{thm:GHGK}
Let $\psi$ be a pure state of 3 parties of arbitrary finite dimensions with $\|\psi\|^2=m$, and let $(a,b,c)$ be its decreasing marginal spectra. For every Horn triple $(I,J,K)$ and for every $l$ and $j\notin I$
\begin{align}
a_1-a_{J+1}+b_K-b_l-c_I+c_l\ &\le m \label{GH}\\
-a_{J+1}+b_j+c_K&\le m\ .\label{GK}
\end{align}
They hold for all partitions $\lambda,\mu,\nu\vdash m$ with $g\ne0$. 
\end{corollary}
 
\begin{proof}
Both are lifts of Eq.~\eqref{eq:seed_1}.
 
To see Eq.~\eqref{GH}, apply Theorem~\ref{thm:lift} to the seed of Eq.~\eqref{eq:seed_1} 
\begin{align}
    \gamma_1+\beta_J\le m+\hat a_J\ ,
\end{align}
with $u=\one_J$, $v=\one_1$, $w=\one_J$ and $\kappa=1$.
Take the Horn triples $(I,J,K)$ on $B$ and $(\{l\},\{1\},\{l\})$ on $C$, $L=\{l\}$ and $L'=I$. Then $\delta=\one_I+\one_l-\one_l-\one_I=0$, so $\bar\kappa=0$, and Eq.~\eqref{eq:FL} becomes
\begin{align}
    a_1-a_{J+1}+b_K-b_l-c_I+c_l\ \le\ m\ .
\end{align}

For Eq.~\eqref{GK}, apply Theorem~\ref{thm:lift} to the seed
\begin{align}
    \beta_1+\gamma_J\le m+\hat a_{J}\ ,
\end{align}
with $u=\one_1$, $v=\one_J$, $w=\one_J$ and $\kappa=1$. Take triple $(\{j\},\{1\},\{j\})$ on $x$, the Horn triple $(I,J,K)$ on $c$, and $L=L'=\emptyset$.
Since $j\notin I$, $\delta=\one_j+\one_I$ is a $0/1$-vector, so $\bar\kappa=1$. The $\alpha$-terms are absorbed by $a_1$. Then Eq.~\eqref{eq:FL} reads 
\begin{align}
    b_j+c_K\le m+a_{J+1}\ .
\end{align}
\end{proof}

Corollary~\ref{thm:GHGK} and Ressayre's inequalities both carry Horn inequalities over to non-zero Kronecker coefficients, but in different ways. Ressayre starts from the face $\lambda_1+\mu_1-\nu_1=n$, on which $g(\lambda,\mu,\nu)= c^{\hat\nu}_{\hat\lambda\hat\mu}$, and shifts the indices of all 3 parties, with $|K|=|I|+|J|$ and length bounds tied to $(e,f)$. Corollary~\ref{thm:GHGK} starts from peeling: it shifts only the peeled party, uses a Horn triple of equal-size sets on the other two, and holds without any restriction on the local dimensions. The families overlap. Eq.~(20) with $i=1$ is the case $I=J=K=\{1\}$, $l=j$ of Eq.~(36), and it also follows from \cite{ressayre2019horninequalitiesnonzerokronecker} with $I=\{1\}$, $J=\{j-1\}$, $K=\{j-1,e+f\}$ after dropping $\nu_{e+f+1}\ge0$.

\section{Dimension 4 and 5}\label{sec:d4}

We already know that all 7 inequalities in dimension 3 follow from Theorem~\ref{thm:G23}.
By computer-assisted, exact calculations we have verified that this also holds for the 50 facet classes of the marginal polytope of dimension 4~\cite{Vergne_2017,bulois2025algorithmcomputekroneckercone}. 11 of these are trivial or first row bounds, 32 are obtained with a single peeling step, and the remaining 7 need 2 peeling steps.
Of the 466 classes in dimension 5~\cite{bulois2025algorithmcomputekroneckercone}, 455 versions are valid in every dimension, derived from at most 2 peeling steps. 20 of these are trivial or first row bounds, 139 are one-step instances and further 296 need 2 peeling steps.
Interestingly, 11 classes are not implied by the inequalities recovered by Dvir's recursion at any depth.

That these inequalities do not follow from Theorem~\ref{thm:lift} can be shown in the following way. In each step, the peeled party’s rank drops by exactly one, while the other 2 parties do not grow. So each step lowers the total rank by at least one. When a party has rank 1, the state is a product state and the other 2 parties have the same spectra. This finite set of inequalities defines a polyhedron. Finding a point outside the set of admissible marginal spectra that is not excluded by the finite list of inequalities derived from Theorem~\ref{thm:lift}, shows that the list is not complete.
The resulting inequalities are checked exactly using only Horn's inequalities for $n=5$ recomputed from the Littlewood--Richardson coefficients.
Using the list of Ref.~\cite{bulois2025algorithmcomputekroneckercone}, for each one of the 11 remaining classes we found a point that is excluded by the corresponding class and therefore outside the quantum marginal polytope, but that is not excluded by our inequalities.
This means that the inequalities derived from Theorem~\ref{thm:lift} are complete for $d\le4$ and incomplete for every $d\ge5$.

\section{Conclusion}

This article establishes an important connection between Horn's problem and the quantum marginal problem of a pure tripartite system. This connection is the key to find a family of inequalities that are valid for arbitrary dimensions. Although not exhaustive in general, these inequalities fully describe the marginal polytope for dimension 2, 3 and 4. But even if they fail to give a complete description of the marginal polytope in dimension 5 or higher, they still recover a large fraction of all inequalities. This work is a decisive contribution to the theoretical understanding of the structure of eigenvalues in the pure tripartite quantum systems. At the same time, the introduced family of inequalities is of practical relevance. It can be efficiently computed in polynomial time and the resulting set gives a outer approximation of the set of allowed marginals. Such relaxations on the quantum set have proven useful in other contexts.

\section{Contributions and Acknowledgments}

SM conceived the main idea and key steps in the proofs. Most of the technical parts of this work was done with Claude Opus 5.5.

SM is thankful to Jens Siewert and Felix Huber for insightful discussions.
SM is a recipient of an APART-MINT Fellowship of the Austrian Academy of Sciences at the Atominstitut of the Technische Universit\"at Wien.

\bibliography{bibfile.bib}

\onecolumngrid
\appendix

\renewcommand{\thesection}{\Alph{section}}
\renewcommand{\thesubsection}{\Alph{section}.\arabic{subsection}}
\renewcommand{\thesubsubsection}{\Alph{section}.\arabic{subsection}.\Roman{subsubsection}}


\end{document}